\documentclass[12pt,oneside,english]{article}
\usepackage{amsthm, amsmath, amssymb, euscript, mathrsfs, MnSymbol, verbatim, enumerate, multicol, multirow, color,babel, geometry, tikz,tikz-cd, tikz-3dplot, tkz-graph, array, hepth, enumitem, hyperref, thm-restate, thmtools, datetime}
\usepackage[numbers,sort&compress]{natbib}
\usepackage[utf8]{inputenc}
\usepackage[T1]{fontenc}
\usepackage[normalem]{ulem}
\usepackage{lipsum}

\numberwithin{equation}{section}
\numberwithin{figure}{section}

\usetikzlibrary{arrows,positioning,decorations.pathmorphing,   decorations.markings, matrix, patterns}

\hypersetup{colorlinks=true}
\hypersetup{linkcolor=black}
\hypersetup{citecolor=black}
\hypersetup{urlcolor=black}

\theoremstyle{plain}
\newtheorem*{thm*}{Theorem}
\newtheorem{thm}{Theorem}[section]

\newtheorem{lem}[thm]{Lemma}

\newtheorem*{cor}{Corollary}

\theoremstyle{definition}
\newtheorem{defn}[thm]{Definition}
\newtheorem*{defn*}{Definition}

\tikzset{
  big arrow/.style={
    decoration={markings,mark=at position 1 with {\arrow[scale=1.5,#1]{>}}},
    postaction={decorate},
    shorten >=0.4pt},
  big arrow/.default=black}

\makeatother

\begin{document}
\date{}
\institution{Northeastern}{\centerline{${}^{1}$Department of Mathematics, Northeastern University, Boston, MA, USA}}
\institution{HarvardPhys}{\centerline{${}^{2}$Department of Physics, Harvard University, Cambridge, MA, USA}}

\title{Characteristic numbers of elliptic fibrations with non-trivial Mordell--Weil groups}
\authors{Mboyo Esole\worksat{\Northeastern}\footnote{e-mail: {\tt j.esole@northeastern.edu}} and Monica Jinwoo Kang\worksat{\HarvardPhys}\footnote{e-mail: {\tt jkang@physics.harvard.edu}} }

\abstract{
We compute characteristic numbers of elliptically fibered fourfolds with multisections or non-trivial Mordell--Weil groups. We first consider the models of type E$_{9-d}$ with $d=1,2,3,4$ whose generic fibers are normal elliptic curves of degree $d$. We then analyze the characteristic numbers of the  $Q_7$-model, which provides a smooth model for elliptic fibrations of rank one and generalizes the E$_5$, E$_6$, and E$_7$-models. Finally, we examine the characteristic numbers of $G$-models with $G=\text{SO}(n)$ with $n=3,4,5,6$ and $\text{G}=\text{PSU}(3)$ whose  Mordell--Weil groups are respectively $\mathbb{Z}/2\mathbb{Z}$ and  $\mathbb{Z}/3 \mathbb{Z}$. In each case, we compute the Chern and Pontryagin numbers, the Euler characteristic, the holomorphic genera, the Todd-genus, the L-genus, the A-genus, and  the eight-form curvature invariant from M-theory. 
}
\maketitle

\tableofcontents

\newpage
\section{Introduction}

 In F-theory, we attach to an elliptic fibration $\varphi: Y\longrightarrow B$ a complex compact Lie group $G$ and a representation $\mathbf{R}$ of the Lie algebra $\mathfrak{g}$ of $G$. 
 Such an elliptic fibration is called a \emph{$G$-model} \cite{Euler}. 
 $G$-models are used to geometrically engineer gauge theories in compactifications of M-theory and F-theory \cite{Bershadsky:1996nh,Vafa:1996xn}.
They also play an important role in the study of superconformal gauge theories \cite{Heckman:2018jxk} and Higgs bundles \cite{Anderson:2017zfm}. 
 $G$-models are typically defined by crepant resolutions of singular Weierstrass models given by Tate's algorithm \cite{Bershadsky:1996nh, Esole.Elliptic,Vafa:1996xn}. 
The geometry of $G$-models provides a rich setting for discussions at  the boundary of string theory, algebraic geometry, arithmetic geometry, singularity theory, combinatorics, and birational geometry.

In the last few years, there was significant progress in the longstanding problem  to  describe all crepant resolutions of  a given $G$-model  given by a singular Weierstrass model and understand the geography of the flops connecting them  \cite{Andreas:2009uf,G2,Euler,SUG,SO4,EKY,EKY2,ES,ESY1,ESY2,EY,Box,Mayrhofer:2014opa}. 
  Another natural step in our understanding of these geometries is to explore further their topology and intersection rings.  
The intersection ring is not invariant under flops as we see already in the study of triple intersection numbers of  elliptically fibered threefolds connected by flops \cite{G2,F4,SUG,SO4,EKY2,ES}. 

 The most convenient invariants to compute are those preserved under crepant birational maps, as they are independent of a choice of a crepant resolution. 
Such invariants include the Euler characteristic and the Hodge numbers \cite{Batyrev.Betti,Kontsevich}. 
In particular, their values for Calabi-Yau fourfolds were subject of conjectured inspired by string dualities \cite{Blumenhagen:2009yv}. There have been major improvements in the literature that such conjectures are now becoming theorems with even more general assumptions \cite{Euler}. In particular, the Euler characteristics of $G$-models of arbitrary dimensions are systematically derived in \cite{Euler}. 

In addition, if the variety is at most of complex dimension four, its Chern numbers, and hence all its characteristic numbers that are rational combinations of Chern numbers, are also invariant under crepant birational maps \cite{Aluffi.IMRN,Chern}. 
Only elliptic fibrations of dimension four or less are considered in \cite{Chern}, since crepant birational fivefolds do not necessarily have the same Chern numbers as illustrated by Goresky and MacPherson in \cite[Example 2, p221]{GoMa}. 

Furthermore, the same techniques allow the determination of many more invariants. Some of the relevant invariants are painful to collect systematically as they depend on a choice of a crepant resolution and the number of such resolutions can increase quickly \cite{EJJN2,EJJN1,Box}. For example, the characteristic numbers of $G$-models that are smooth fourfolds obtained by crepant resolutions of Weierstrass models have been determined recently in \cite{Chern}. The key to all these developments is a new pushforward theorem that streamlines computations in the intersection ring of a blowup with a center that is a smooth complete intersection \cite{Euler}.

The present paper is a follow-up to \cite{Chern}, where we computed characteristic numbers of $G$-models of complex dimension four for $G=$ SU($n$) ($n=2,3,4,5,6,7$), Spin($7$), Spin($8$), Spin($10$),  G$_2$, F$_4$, E$_6$, E$_7$, and  E$_8$. Each group considered in \cite{Chern} are simply-connected, which implies that the Mordell--Weil group of the generic fiber of the associated elliptic fibration is trivial. We aim to extend the results of \cite{Chern} to elliptic fibrations with multisections or a non-trivial Mordell--Weil group.

Elliptic fibrations with multisections were first studied in string theory in \cite{Berglund:1998va,Bershadsky:1998vn,Klemm:1995tj} as complete intersections in weighted  projective spaces. They were generalized to complete intersections in projective bundles in \cite{AE2,EFY}. Elliptic fibrations with non-trivial Mordell--Weil groups were studied heavily in F-theory. For examples, see \cite{Aspinwall:1998xj,Baume:2017hxm,Braun:2011zm,Braun:2014qka,Braun:2013nqa,Cvetic:2015moa,Grassi:2012qw,Grimm:2010ez,Krause:2011xj,Mayrhofer:2014opa,Mayrhofer:2014laa,Mayrhofer:2014haa,Oehlmann:2016wsb} and  \cite[\S 7]{Weigand:2018rez} for a review.

In this paper, we will focus on the generic models with multisections having trivial Lie algebra $\mathfrak{g}$ with Mordell--Weil groups $\mathbb{Z}/2\mathbb{Z}$, $\mathbb{Z}/3\mathbb{Z}$, and $\mathbb{Z}$. Since the first model has a gauge group SO($3$), we also include in our analysis the $G$-models with $G$=SO($n$) for $n=3,4,5,6$; they all  have the Mordell--Weil group $\mathbb{Z}/2\mathbb{Z}$ and are related to each other by
base changes \cite{SO,SO4}. We also examine the $G$-model with $G$=PSU($3$), which is the generic model with a Mordell--Weil group $\mathbb{Z}/3\mathbb{Z}$ \cite{Aspinwall:1998xj}. Finally, we consider the case of elliptic fibration of rank one. A generic model of an elliptic fibration of rank one was introduced in F-theory by Morrison and Park \cite{Morrison:2012ei}; a smooth model in the birational class of the Morrison-Park model is given by the Q$_7$-model introduced in \cite{EKY}, which generalizes a model introduced by Cacciatori, Cattaneo and Geemen in \cite{Cacciatori:2011nk} sharing the same Jacobian with the Morrison-Park model.

In particular, we determine the characteristic numbers of elliptic fourfolds $\varphi: Y\longrightarrow B$ of the following three types of elliptic fibrations:
 \begin{enumerate}
 \item The generic fiber of $Y$ is  a genus-one normal curve of degree $d$ for $d=1,2,3,4,5$. Such elliptic fibrations are called models of  type E$_{9-d}$.
They are defined in \cite{AE2,EFY}, where they are used to provide strong coupling regimes of several systems of intersecting branes and orientifolds preserving a fraction of supersymmetry and satisfying the tadpole constraints \cite{CDE,AE1,AE2}.  

 \item $Y$ is an elliptic fibration of rank one. In this case, we use the smooth model Q$_7$($\mathscr{L},\mathscr{M})$ introduced in \cite{EKY}, and a generalization of the model of \cite{Cacciatori:2011nk}, birational to the model of Morrison-Park \cite{Morrison:2012ei}.

 \item $Y$ is a $G$-model with $G=\text{SO}(n)$ with $n=3,4,5,$ or $6$ or with G=PSU($3$). The SO($3$)-model is the generic elliptic fibration with a Mordell--Weil group $\mathbb{Z}/2\mathbb{Z}$  \cite{Aspinwall:1998xj,SO}. 
The SO($3$)-model is defined by $b_6=0$ and has a fiber of type I$_2$ over the generic point of $b_4$. The SO($5$)-model and SO($6$)-model are both derived from the SO($3$)-model by the base change $b_4\to t^2$, which replaces the fiber I$_2$ by an I$_4$ \cite{SO}. The model is an SO($5$)  (resp. SO($6$) )  when the generic fiber over $V(t$) is of type I$_4^{\text{ns}}$ (resp.  I$_4^{\text{s}}$) \cite{SO}. 
The SO($4$)-model is a collision of type $A_1+A_1$ with a Mordell--Weil group $\mathbb{Z}/2\mathbb{Z}$ \cite{EKY2}. Such a collision is not uniquely defined since there are several Kodaira fibers whose dual graph is $\widetilde{\text{A}}_1$. The collisions of kodaira fibers describing the SO($4$)-model that have crepant resolutions are I$_2^{ns}$+I$_2^{ns}$, I$_2^{ns}$+I$_2^s$, I$_2^s$+I$_2^s$, III+I$_2^{ns}$, and III+III \cite{SO4}.
The PSU($3$)-model has a fiber of type I$_3^s$ with a Mordell--Weil group $\mathbb{Z}/3\mathbb{Z}$ \cite{Aspinwall:1998xj}. 
\end{enumerate}

\begin{table}[htb]
\begin{center}
\begin{tabular}{ |c|c |c  | }
\hline
Type & Zero scheme of a section of & Ambient space \\
\hline
Q$_7$ & $\mathscr{O}(3)\otimes \pi^* \mathscr{L}^{\otimes 2}\otimes \pi^* \mathscr{M}$  &  $\mathbb{P}(\mathscr{O}_B\oplus\mathscr{M}\oplus \mathscr{L})$\\
\hline
E$_5=D_5$ & $\mathscr{O}(2)\otimes \pi^* \mathscr{L}^{\otimes 2}$,  $\mathscr{O}(2)\otimes \pi^* \mathscr{L}^{\otimes 2}$  &  $\mathbb{P}(\mathscr{O}_B\oplus\mathscr{L}\oplus \mathscr{L}\oplus\mathscr{L})$\\
\hline
E$_6$& $\mathscr{O}(3)\otimes \pi^* \mathscr{L}^{\otimes 3}$  & $\mathbb{P}(\mathscr{O}_B\oplus\mathscr{L}\oplus \mathscr{L})$\\
\hline
  E$_7$&   $\mathscr{O}(4)\otimes \pi^* \mathscr{L}^{\otimes 4}$  &  $\mathbb{P}_{1,1,2}(\mathscr{O}_B\oplus\mathscr{L}\oplus \mathscr{L}^{\otimes 2})$ \\
 E$^{'}_7$& $\mathscr{O}(3)\otimes \pi^* \mathscr{L}^{\otimes 4}$  & $\mathbb{P}(\mathscr{O}_B\oplus\mathscr{L}\oplus \mathscr{L}^{\otimes 2})$\\
\hline
E$_8$&  $\mathscr{O}(3)\otimes \pi^* \mathscr{L}^{\otimes 6}$  & $\mathbb{P}(\mathscr{O}_B\oplus\mathscr{L}^{\otimes 2}\oplus \mathscr{L}^{\otimes 3})$\\
 \hline

\end{tabular}
\end{center}
\caption{The E$_8$-model is the usual Weierstrass model defined by Deligne and Tate. 
The E$_6$ and E$_7$-models are defined in \cite{AE1} while the D$_5$-model is defined in \cite{EFY}.  The E$^{'}_7$ and Q$_7$ model are respectively introduced in \cite{Cacciatori:2011nk} and \cite{EKY}. The Q$_7$-model specializes to E$_6$ and E$_7^{'}$ when $\mathscr{M}$ is $\mathscr{L}$ and $\mathscr{L}^{\otimes 2}$ respectively. 
\label{Table.Models}}
\end{table}

\subsection{Characteristic numbers considered}
Following \cite{Chern}, for each of the elliptic fibrations, we compute  the following six rational Chern and Pontryagin numbers:
\begin{enumerate}
\item The Chern numbers 
\begin{equation}
\text{$\int_Y c_1(TY)^4$, $\int_Y c_1(TY)^2 c_2 (TY)$, $\int_Y c_1(TY) c_3 (TY)$, $\int_Yc_2^2(TY)$,  and $\int_Y c_4 (TY).$}
\end{equation}
\item The holomorphic genera  $\chi_p(Y) = \sum_{q=0}^n (-1)^q h^{p, q}(Y)$ \cite{Klemm:1996ts}:
  \begin{equation}
  \begin{aligned}
  \chi_0(Y)&=\int_Y \mathrm{Td}(TY)=\frac{1}{720}\int_Y (-c_4 + c_1 c_3 +3 c_2^2 + 4 c_1^2 c_2 -c_1^4),\\
  \chi_1(Y) & = \frac{1}{180}\int_Y (-31c_4 -14 c_1 c_3 +3 c_2^2 + 4 c_1^2 c_2 -c_1^4),\\
  \chi_2(Y) & = \frac{1}{120}\int_Y (79c_4 -19 c_1 c_3 +3 c_2^2 + 4 c_1^2 c_2 -c_1^4).
  \end{aligned}
  \end{equation}

\item The  Pontryagin numbers $\int_Y p_2(TY)$ and $\int_Y p_1^2(TY)$, where the Pontryagin classes $p_1(TY)$ and $p_2(TY)$ are defined as
 \begin{equation}
 \begin{aligned}
 p_1(TY) &=c_1^2(TY) -2c_2(TY), \\
  p_2(TY) &= c_2^2(TY)-2 c_1(TY) c_3(TY)+2c_4(TY).
 \end{aligned}
 \end{equation}
\item The Hirzebruch signature of a fourfold, 
 \begin{equation}
\sigma(Y)=\frac{1}{45}\int_Y \Big(7p_2(TY) -p_1^2(TY)\Big).
\end{equation}
The signature is the degree of the Hirzebruch $L$-genus. 
\item The $\hat{\text{A}}$-genus of a fourfold,
\begin{equation}
\begin{aligned}
\int_Y \hat{\text{A}}_2(TY)&=\frac{1}{5760}\int_Y\Big(7p_1^2(TY)-4p_2(TY)\Big).
\end{aligned}
 \end{equation}
By the Atiyah--Singer theorem, if the fourfold $Y$ is a spin manifold,  the degree of  $ \hat{\text{A}}_2$ gives the index of the Dirac operator on  $Y$. 
\item 
 We also compute the following invariant that plays an important role in many questions of anomaly cancellations in type IIA, M, and F-theory \cite{Sethi:1996es,Vafa:1995fj}:
\begin{equation}
X_8(Y)=\frac{1}{192 }\int_Y \Big( p_1^2(TY) -4p_2(TY)\Big).
\end{equation}
\end{enumerate}

Even though Chern numbers other than the Euler characteristic are not topological invariants, some are invariant under crepant birational maps as proven by Aluffi.

\begin{thm}[Aluffi, { \cite[page 3368]{Aluffi.IMRN}}]\label{Thm:Al}
For two nonsingular $n$-dimensional complete varieties $X$ and $Y$ connected by a crepant birational map,
$$
\int_X c_1(TX)^i c_{n-i}(TX)=\int_Y c_1(TY)^i c_{n-i}(TY), \quad i=0, 1, \dots, n.
$$
\end{thm}

The following theorem asserts that the Chern numbers of fourfolds are  $K$-equivalence  invariants.  The proof follows from  Theorem \ref{Thm:Al} of Aluffi and the birational invariance of the Todd-genus. 
\begin{thm}[Esole--Kang, {\cite{Chern}}]\label{Thm:TheInvariance}
The Chern and Pontryagin numbers of an algebraic variety of complex dimension four are $K$-equivalence invariants.
\end{thm}
The Chern number $\int_Yc_1(TY)^2 c_2(TY)$, the $\hat{A}$-genus, and the Todd-genus (the holomorphic Euler characteristic) are invariants of the choice of $G$. 
They can all be expressed as invariants of the divisor $W$ defined by the vanishing locus of a smooth section of $\mathscr{L}$. 

By expressing $\chi_0$ and $\hat{A}$ in terms of Chern numbers and using the identity $\int_Y c_1^4=0$, which holds for 
any elliptic fibration related to a Weierstrass model by a crepant birational map  (see Theorem \ref{thm.1.4}), we get the following expressions of $c_1^2 c_2$ and $c_1 c_3$:
\begin{align}
\int_Y c_1^2c_2=96(\chi_0(Y)-\hat{A}(Y)),\quad 
\int_Y c_1 c_3 = 384\hat{A}(Y)+ 336\chi_0(Y)+\chi(Y)- 3 \int_Y c_2^2.
\end{align}
For an elliptic fibration that is a crepant birational map away from a Weierstrass model, this shows that $\int_Y c_1^2 c_2$ gives the same value as a smooth Weierstrass model with the same  fundamental line bundle $\mathscr{L}$. 
It is therefore enough to compute only $\int_Y c_2^2$ and the Euler characteristic $\chi(Y)=\int_Y c_4$ to know all the Chern numbers.

\subsection{The Calabi-Yau fourfold case}\label{sec:CY4}
In the case of Calabi-Yau fourfolds, knowing the Euler characteristic is enough to also compute other invariants such as the Chern number $c_2(TY)^2$ as a function of the Euler characteristic. 
\begin{thm}\label{Thm2.CY4}
 The Chern numbers and Pontryagin numbers of a  Calabi-Yau fourfold depend  only on its Euler characteristic. 
\end{thm}
In the Calabi-Yau fourfold case, all Chern and Pontryagin numbers  can be read from the computation of the Euler characteristic of $G$-models in \cite{Euler} as  we have \cite{Klemm:1996ts,Sethi:1996es}
  \begin{equation}
  \begin{aligned}
\int_Y  c_2^2(TY) &=480+\frac{1}{3}\chi(Y),\quad
  \sigma =32+\frac{1}{3}\chi(Y),\\
  \chi_0&=2 ,\quad
  \chi_1  =8-\frac{1}{6} \chi(Y),\quad
  \chi_2  =12+ \frac{2}{3}\chi(Y),\\
  X_8&=-\frac{1}{24}\chi(Y),\quad
 \frac{1}{5760} \int_Y \hat{\text{A}}_2 = 2.
   \end{aligned}
  \end{equation}
For $G$-models, the Euler characteristic of a Calabi-Yau fourfold is given in Table 10 of \cite{Euler}, which we reproduce here for completeness. We notice that the same table can be obtained from the Euler characteristic $c_4(TY)$ after putting $L=c_1$, which is the condition that ensures the triviality of the canonical class of $Y$.

 The rest of the paper is organized as follows. In Section \ref{sec:modelsconsidered}, we introduce all the models considered in this paper. First we explain E$_8$, E$_7$, E$_6$, D$_5$, and Q$_7$-models in Section \ref{sec:genusone} and \ref{sec:Enmodels}. We then explain the $G$-models considered in this paper with $G=\text{SO}(n)$ for $n=3,4,5,6$ in Section \ref{sec:SO(n)} and with $G=\text{PSU}(3)$ in Section \ref{sec:PSU3}. We summarize necessary pushforward theorems used and introduce theorems for the models considered in this paper in Section \ref{sec:push}. Finally, we present our  results in Section \ref{sec:results}.

\section{
Descriptions of the elliptic fibrations considered
}
 \label{sec:modelsconsidered}
 
In this section, we introduce the elliptic fibrations considered in this paper. Namely, the elliptic fibrations of type E$_8$, E$_7$, E$_6$ \cite{AE1}, D$_5$ \cite{EFY}, and Q$_7$ \cite{EKY}, the SO($n$)-models for $n=3,4,5,6$ \cite{SO,SO4}, and the PSU($3$)-model \cite{Aspinwall:1998xj}.

\subsection{Genus-one normal curves of degree $d$} \label{sec:genusone}
The Riemann-Roch theorem for curves famously implies that a genus-one curve with a rational point can be expressed as the zero scheme of a Weierstrass equation \cite{MumfordSuominen}. 
Using some cohomology and base change techniques, the result can also be extended to fibrations of genus-one curves with a rational section \cite{Formulaire, MumfordSuominen}. 
A divisor of degree $d\geq 3$ defines an embedding of the genus-one curve into $\mathbb{P}^{d-1}$, such a curve is a called a {\em genus-one normal curve of degree $d$}.  For a review on genus-one normal curves, see \cite{Hulek} and \cite[\S3]{Fisher}. 
Some of the historical most famous families of elliptic curves are genus-one normal curves of degree $1$, $2$, $3,$ and $4$. These are precisely the genus-one normal curves that are hypersurfaces or complete intersections in a (weighted) projective space \cite{Hulek}. 
A rational point is a divisor of degree one on the curve. By the Riemann-Roch theorem, a divisor of degree one on a genus-one curve yields a Weierstrass equation \cite{MumfordSuominen}.
A divisor of degree two gives a quartic curve that can be  embedded as a quartic in a weighted projective space $\mathbb{P}_{1,1,2}$. 
Since the weighted projective space $\mathbb{P}_{1,1,2}$ is mapped to a singular quadric surface by a map of degree two, we can also think of quartic in $\mathbb{P}_{1,1,2}$ as a particular case of a complete intersection of quadric surfaces in $\mathbb{P}^3$.
 A divisor of degree three gives a projective cubic curve in $\mathbb{P}^2$. 
A divisor of degree four gives the complete intersection of two quadrics surfaces in  $\mathbb{P}^3$.  

In the string theory literature, an elliptic fibration whose generic fiber is a genus-one normal curve of degree $d=1,2,3,4$ is called a model of type  E$_{9-d}$ (we recall that E$_5$ is isomorphic to D$_5$) \cite{Klemm:1996ts,AE2,EFY}. 
They are historically named after some del Pezzo surfaces in which they can be embedded as hyperplane divisors.

A del Pezzo surface is a smooth surface with an ample anticanonical divisor \cite{DelPezzo}.  It follows that the anticanonical model of a del Pezzo surface is isomorphic to the del Pezzo surface. 
The degree of a del Pezzo surface is the square of its canonical class.
In string theory,  a del Pezzo surface of degree $d$ is usually called a dP$_{9-d}$. 
If $X$ is a del Pezzo surface of degree $d\geq 3$, then  its anticanonical model is a smooth surface of degree $d$ in $\mathbb{P}^d$. 
If $X$ is a del Pezzo surface of degree one, then its anticanonical model is  a hypersurface of degree six in the weighted projective space $\mathbb{P}_{1,1,2,3}$. 
If $X$ is a del Pezzo surface of degree two, then its anticanonical model is a hypersurface of degree four in $\mathbb{P}_{1,1,2,3}$. 
A del Pezzo surface of degree eight is isomorphic to a quadratic Veronese embedding of a quadric in $\mathbb{P}^3$.  
A  del Pezzo surface of degree $3\leq d\leq 8$ is isomorphic to  $\mathbb{P}^2$ blown up at $(9-d)$ points in general position.

In the late 1960s, Manin discovered that a del Pezzo surface $X$ of degree $d\leq 6$ is associated with a root system E$_{9-d}$   \cite{Manin}: the automorphism group of the incidence
graph of the exceptional curves on $X$ is the Weyl group W(E$_{9-d}$) of the Lie algebra of type E$_{9-d}$. By Bertini's theorem, a general hyperplane section on a del Pezzo surface of degree $d$ defines a smooth (normal) elliptic curve of degree $d$.  
Such an elliptic curve is said to be of type E$_{9-d}$ when $d=1,2,3,4,5$ \cite{Klemm:1996ts,AE2,EFY}.

\begin{table}[htb]
\begin{center}
\begin{tabular}{ |   c |  c   |c|}
\hline
  Weierstrass model ($E_8$-type)   & $y^2=x^3+ f x  + g $ 
  \\
Legendre family       & $y^2= x(x-1)(x-\lambda )$ 
 \\
 Hesse family  (E$_6$-type)       & $x^3+y^3+ 3\mu x y +1 =0$ 
 \\
  Jacobi's quartic (E$_7$-type)   & $y^2=x^4 + 2\kappa x^2 + 1$ 
  \\
  Jacobi's  intersection (D$_5$-type)& $x^2+y^2-z^2=k x^2+ w^2-z^2=0$ 
   \\
   \hline
  Newton's  cubic hyperbola ($Q_7$ type)& 
  $yx^2 + Ax = By^3 + Cy^2 + Dy + E$ 
  \\
  \hline 
\end{tabular}
\end{center}
\caption{Some classic families of elliptic curves. The Newton's cubic hyperbola is also than most of these curves but has not been used until recently in \cite{EKY}. The Weierstrass model was also in Newton's list of cubics. }
\end{table}

\subsection{E$_8$, E$_7$, E$_6$, D$_5$, and Q$_7$-models} \label{sec:Enmodels}
The model of type E$_8$ is the usual Weierstrass model \cite{Formulaire, MumfordSuominen}.  
 The E$_8$-model is often described in string theory by a curve of degree six in the weighted projective space $\mathbb{P}_{1,2,3}$, specially in toric constructions. 
The Weierstrass equation appears automatically in this form, but it has a major drawback as the ambient space has singularities which require much care, especially when doing intersection theory. 
 In turn, we prefer using the projective bundle $\mathbb{P}[\mathscr{O}_B\oplus\mathscr{L}^{\otimes 2}\oplus\mathscr{L}^{\otimes 3}]\to B$ \cite{Euler}. 

Models of type  E$_7$, E$_6$, and E$_5=\text{D}_5$ were first realized as hypersurfaces in weighted projective spaces \cite{Berglund:1998va,Bershadsky:1998vn,Klemm:1995tj,Klemm:1996ts}. 
They were generalized to arbitrary bases by using projective bundles in \cite{AE2, EFY}. In  \cite{EKY}, we introduced a new family modeled after  Newton's cubic hyperbola generalizing the family introduced in \cite{Cacciatori:2011nk}. 
Another generalization would be to consider elliptic surfaces whose generic fiber are defined as hypersurfaces in toric surfaces \cite{Grassi:2012qw,Braun:2013nqa,Klevers:2014bqa}.  
There are in total sixteen two dimensional reflexive polyhedra. By the adjunction theorem, each of them can define a genus one curve as an anticanonical divisor.

For example, the generic equation of an elliptic fibration with Mordell--Weil group $\mathbb{Z}/2\mathbb{Z}$ has a singularity that has a crepant resolution defined by blowing up the point $x=y=0$. 
This turns the $\mathbb{P}^2$-bundle  $X_0=\mathbb{P}(\mathscr{O}_B\oplus\mathscr{L}^{\otimes 2}\oplus \mathscr{L}^{\otimes 3})$ into a fibration of dP$_1$ surfaces. 
Interestingly, the discriminant locus has a fiber of type I$_2$, and therefore the gauge group is SO($3$) as we have to take the quotient of SU($2$) by a central  $\mathbb{Z}/2\mathbb{Z}$.

\begin{enumerate}
\item The Weierstrass model is a zero scheme of  a section of $\mathscr{O}(3)\otimes \pi^* \mathscr{L}^{\otimes 6}$  in $\mathbb{P}(\mathscr{O}_B\oplus\mathscr{L}^{\otimes 2}\oplus \mathscr{L}^{\otimes 3})$.
This is the usual Weierstrass model given for example by Tate's form. 

\item The elliptic fibration of type E$_7$ is  a zero scheme of  a section of $\mathscr{O}(4)\otimes \pi^* \mathscr{L}^{\otimes 4}$  in $\mathbb{P}_{1,1,2}(\mathscr{O}_B\oplus\mathscr{L}\oplus \mathscr{L}^{\otimes 2})$ \cite{AE2}. 
Its generic fiber corresponds to a fiber of the Jacobi quartic.
The elliptic fibration of type E$_7$ can be written as a zero scheme of  a section of $\mathscr{O}(3)\otimes \pi^* \mathscr{L}^{\otimes 4}$  in $\mathbb{P}(\mathscr{O}_B\oplus\mathscr{L}\oplus \mathscr{L}^{\otimes 2})$ as proven in \cite{Cacciatori:2011nk}. 

\item The elliptic fibration of type E$_6$ is  a zero scheme of  a section of $\mathscr{O}(3)\otimes \pi^* \mathscr{L}^{\otimes 3}$  in $\mathbb{P}(\mathscr{O}_B\oplus\mathscr{L}\oplus \mathscr{L})$ \cite{AE2}. 
Its generic fiber has a triple section. 

\item The elliptic fibration of type D$_5$ is a zero scheme of  two  sections of   $\mathscr{O}(2)\otimes \pi^* \mathscr{L}^{\otimes 2}$  in $\mathbb{P}(\mathscr{O}_B\oplus\mathscr{L}\oplus \mathscr{L}\oplus\mathscr{L})$ \cite{EFY}.
Its generic fiber corresponds to a Jacobi complete intersection (Intersection of two quadrics in $\mathbb{P}^3$.
\item The elliptic fibration of type Q$_7$($\mathscr{L},\mathscr{M}$) is a zero scheme of  a section of  $\mathscr{O}(3)\otimes \pi^* \mathscr{L}^{\otimes 2}\otimes \pi^* \mathscr{M}$  in $\mathbb{P}(\mathscr{O}_B\oplus\mathscr{M}\oplus \mathscr{L}^{\otimes 2})$ \cite{EKY}.
Its generic fiber is of the type of ``cubic hyperbola'' that appeared in Newton's classification of cubics. It extrapolates between the  E$_7^{'}$ family (at $\mathscr{M}=\mathscr{L}^{\otimes 2}$) and the family introduced by  Cacciatori, Cattaneo, and van Geemen  in \cite{Cacciatori:2011nk}, which corresponds to $\mathscr{M}=\mathscr{O}_B$. 
\end{enumerate}

\subsection{Newton's cubic hyperbola, Morrison-Park model, and the $Q_7$-model}\label{Sec:Q7}
Morrison and Park have popularized in the F-theory literature the Jacobian of an elliptic fibration of rank one  \cite{Morrison:2012ei}.
The model is also well-known to number theorists since it corresponds to a  genus-one normal curve of degree  two  in the special case where the degree-two divisor splits into two rational points without having to introduce a field extension. 
 A smooth model sharing the same Jacobian  appeared in the F-theory literature a year earlier in \cite{Cacciatori:2011nk} where Cacciatori, Cattaneo, and Geemen studied tadpole cancellation conditions in the spirit of \cite{AE2,AE1}  in a new projective bundle.  The Morrison--Park model  \cite{Morrison:2012ei} matches the Jacobian of the new model of Cacciatori, Cattaneo, and Geemen  \cite{Cacciatori:2011nk} after the change of variables 
\begin{equation}
c_0 \to -\frac{1}{6} c_0, \   
 c_1 \to  -c_1, \  
 c_2\to -6  c_2,\  
 c_3 \to -36 c_3, \   
 b \to  6 \sqrt{6} b_2.
\end{equation}
Both models can be traced directly to  Newton's cubic hyperbola as explained in \cite{EKY}, as they can both be written as the following cubic hyperbola:
\begin{equation}
y x^2+b z^2 x=  c_0 y^3 + c_1 y^2 z + c_2 y z^2 + c_3 z^3.
\end{equation}
We call this model a $Q_7$-model because the Newton's polygon of the defining equation is a quadrilateral  with seven points on its boundary and a unique  interior point (ensuring that it describes a curve of arithmetic genus one). 
This Newton's polygon is one of the sixteen reflexive polygons of degree one  \cite{Batyrev.1985,Koelman}. 
Since the genus-one fibration has a rational section at $z=y=0$, it is birational to a Weierstrass model, which can be determined using the results of \cite{Tate:2005}. 
The model of type  $Q_7(\mathscr{L},\mathscr{M})$ is a hypersurface in a projective bundle generalizing the elliptic fibration introduced in \cite{Cacciatori:2011nk}. 
We introduce two line bundles $\mathscr{L}$ and $\mathscr{M}$ and define the ambient space as the projective bundle  $X_0=\mathbb{P}(\mathscr{O}_B\oplus\mathscr{M}\oplus \mathscr{L})$ \cite{EKY}.
We denote by $\mathscr{O}(1)$ the dual of the tautological line bundle of $X_0$. 
 The projective coordinates $[z:y:x]$ are such that  $z,y,x$ are respectively sections of $\mathscr{O}(1)$, $\mathscr{O}(1)\otimes \mathscr{L}$, and $\mathscr{O}(1)\otimes \mathscr{M}$. The defining equation is a section  of the 
line bundle $\mathscr{O}(3)\otimes \pi^* \mathscr{L}^{\otimes 2}\otimes \pi^* \mathscr{M}$. The coefficient $c_i$ is a  section of $\mathscr{L}^{\otimes 2}\otimes\mathscr{M}^{\otimes(-2+i)}$ and $b$ is a section of $\mathscr{L}\otimes\mathscr{M}$.  
At $z=0$, we also have the degree-two divisor $z=x^2-c_0 y^2=0$. Thus, it follows from Riemann-Roch that the $Q_7$-model is also birational to a quartic model \cite{EKY}. The quartic model can be derived by multiplying both side of the defining equation by $y$ and introduce the new variable $u=x y$. 
This gives the quartic model
\begin{equation}
u(u+b z^2)=  y(c_0 y^3 + c_1 y^2 z + c_2 y z^2 + c_3 z^3).
\end{equation}
The variable $u$ is a section of $\mathscr{O}(2)\otimes\pi^* \mathscr{L}\otimes \pi^*\mathscr{M}$. 
 The ambient space is the weighted projective bundle 
$\pi:\mathbb{P}_{1,2,1}(\mathscr{O}_B\oplus \mathscr{U}\oplus \mathscr{L})\longrightarrow B$ with $\mathscr{U}=\mathscr{L}\otimes\mathscr{M}$. The defining equation is a section of $\mathscr{O}(4)\otimes \pi^* \mathscr{U}^{\otimes 2}$.
The quartic equation has  double point singularities at $u=b=y=c_3=0$.

\subsection{$G$-models for $G=\text{SO}(n)$ with $n=3,4,5,6$}\label{sec:SO(n)}
The SO($3$)-model is the generic case of a Weierstrass model with a $\mathbb{Z}/2\mathbb{Z}$ torsion \cite{Aspinwall:1998xj}:
\begin{equation}
 y^2z= x(x^2 + a_2 x z+ a_4 z^2).
\end{equation}
There is a fiber of type I$_2$ over the generic point of $V(a_4)$. The Weierstrass models for the SO($3$), SO($5$), and SO($6$)-models are summarized in Table \ref{Table:Weierstrass}. The SO($5$) and SO($6$)-models are derived in \cite{SO}.

The SO($5$)-model is obtained from the SO($3$)-model  via a base change that converts the section $a_4$ to a perfect square $a_4=t^2$ where $t$ is a section of $\mathscr{L}^{\otimes 2}$. 
The generic fiber over the generic point of $V(t)$ is of type I$_4^{\text{ns}}$. Moreover, $a_2$ must be nonzero, otherwise the Mordell--Weil group becomes $\mathbb{Z}/2\mathbb{Z}\times \mathbb{Z}/2\mathbb{Z}$, which in turn implies $G = 
\text{PSO}(5)$ rather than SO($5$). 

The SO($6$)-model is subsequently derived from the SO($5$)-model by requiring that the section $a_2$ is a perfect square modulo $t$ so that the fiber type over $V(t)$ is of type I$_4^{\text{s}}$.

The SO($4$)-model is defined by a collision of two Kodaira fibers with dual graph $\widetilde{\text{A}}_1$ with a Mordell--Weil group $\mathbb{Z}/2\mathbb{Z}$ \cite{SO4}. The simplest SO($4$)-model is realized by III+III, the collision of two fibers of type III, in an elliptic fibration with a Mordell--Weil group $\mathbb{Z}/2\mathbb{Z}$. For all the possible realizations of SO($4$)-models, see Table \ref{Table:Weierstrass}.
 We give a crepant resolution for each model in Table \ref{tab:blowupcenters}.

\begin{table}[htb]
 \begin{center}
  $
\begin{array}{|c| c | c | l |}
\hline 
 \text{Group} & \text{Kodaira fibers}  &\text{Mordell--Weil}&  \text{Weierstrass model}  \\
\hline
\text{SO}(3) &   \text{I}_2^{\text{ns}}, \     \text{III}& \mathbb{Z}/ 2\mathbb{Z} & y^2z= x(x^2 + a_2 x z+ s z^2)\\
\hline 
 \multirow{5}{*}{SO($4$)}& \text{I}^{\text{ns}}_2+\text{I}^{\text{ns}}_2\to \text{I}^{\text{ns}}_4&   \multirow{5}{*}{$\mathbb{Z}/ 2\mathbb{Z}$} & y^2z=x^3+a_2x^2z+stxz^2  \\
&\text{I}^{\text{ns}}_2+\text{I}^{\text{s}}_2\to \text{I}^{\text{s}}_4 & & y^2z+a_1 xyz=x^3+\widetilde{a}_2 tx^2z+stxz^2 \\
 &\text{I}^{\text{s}}_2+\text{I}^{\text{s}}_2\to \text{I}^{\text{s}}_4 &   & y^2z+a_1 xyz=x^3+ \widetilde{a}_2 st x^2 z+stxz^2 \\
  &\text{III}+\text{I}^{\text{ns}}_2\to 1-2-1  & & y^2z=x^3+\widetilde{a}_2 s x^2 z+stxz^2 \\
 &\text{III+III}\to 1-2-1 &  &  y^2z=x^3+\widetilde{a}_2 st x^2 z+stxz^2 \\
\hline
\text{SO}(5) &  \text{I}_4^{\text{ns}} &  \mathbb{Z}/ 2\mathbb{Z}&  y^2z= x(x^2 + a_2 x z+ t^2 z^2)\\
\hline
\text{SO}(6)  &  \text{I}_4^{\text{s}}&  \mathbb{Z}/ 2\mathbb{Z}& y^2z+ a_1 y x  z= x(x^2 + m t  x z+ t^2z^2) ~~~ m\neq 0, \pm 2\\
\hline
\text{PSU}($3$) & \text{I}_3^{\text{s}}& \mathbb{Z}/3\mathbb{Z} & y^2z+ a_1 x y z +s yz^2=x^3\\
\hline
\end{array}
$
\end{center}
\caption{Weierstrass equations for the SO($n$)-models for $n=3,4,5,6$  \cite{SO,SO4} and the PSU($3$)-model \cite{Aspinwall:1998xj}. For the SO($3$)-model, the divisor supporting the group is $S=V(s)$, where $s$ is a section of $\mathscr{L}^{\otimes 4}$. For the SO($5$) and SO($6$)-models, the divisor supporting the gauge group is $T=V(t)$, where $t$ is a section of $\mathscr{L}^{\otimes 2}$. 
 For SO($6$), $m$ is a constant number different from $0$ and $\pm 2$ \cite{SO}.
 For the SO($4$)-model, the gauge group is semi-simple and we have two divisors $S=V(s)$ and $T=V(t)$, whose classes satisfy the relation $S+T=4L$. 
For the PSU($3$)-model, the divisor supporting the group is $S=V(s)$, where $s$ is a section of $\mathscr{L}^{\otimes 3}$.
}
\label{Table:Weierstrass}
\end{table}

\subsection{PSU($3$)-model} \label{PSU(3)model}\label{sec:PSU3}
The PSU($3$)-model is the generic Weierstrass model with torsion $\mathbb{Z}/3\mathbb{Z}$ \cite{Aspinwall:1998xj}:
\begin{equation}
y^2z+ a_1 x y z +s yz^2=x^3.
\end{equation}
Its discriminant  is 
\begin{equation}
\Delta=\frac{1}{16}s^3 (27 s -a_1^3).
\end{equation}
The model has a singular fiber of type I$_3^{\text{s}}$ over $V(s)$ where $s$ is a smooth section of the line bundle $\mathscr{L}^{\otimes 3}$. 

A crepant resolution is an embedded resolution defined by a sequence of blowups with smooth centers. We denote the blowup $X_{i+1}\longrightarrow X_i$ along the ideal $(f_1,f_2,\ldots,f_n)$ with exceptional divisor $E$ as:
$$\begin{tikzcd}[column sep=2.4cm]X_i \arrow[leftarrow]{r} {\displaystyle (f_1,\ldots, f_n|E)}  & X_{i+1}\end{tikzcd}.$$
A crepant resolution of a PSU($3$)-model is given by the following sequence of blowups 
\begin{equation}
 \begin{tikzpicture}
	\node(X0) at (0,0){$X_0$};
	\node(X1) at (2.5,0){$X_1$};
	\node(X2) at (5,0){$X_2.$};
	\draw[big arrow] (X1) -- node[above,midway]{$(x,y,s|e_1)$} (X0);	
	\draw[big arrow] (X2) -- node[above,midway]{$(y,e_1|e_2)$} (X1);		
 \end{tikzpicture}
 \end{equation}

\section{Pushforward theorems}
\label{sec:push}
In this section, we prove the pushforward theorems that are the heart of our  computations. 
Throughout this paper, we work over the field of complex numbers. 
A variety is a reduced and irreducible algebraic scheme. 
 We denote the vanishing locus of the sections $f_1, \ldots, f_n$ by $V(f_1, \ldots, f_n)$. 
  The tangent bundle of a variety $X$ is denoted by $TX$ and the  normal bundle of a subvariety $Z$ of a  variety $X$ is denoted by  $N_Z X$. 
 Let $\mathscr{V}\rightarrow B$ be a vector bundle over a variety $B$. We denote the by $\mathbb{P}(\mathscr{V})$ the projective bundle of lines in  $\mathscr{V}$.
 The Chow group $A_*(X)$ of a nonsingular variety $X$ is the group of divisors modulo rational equivalence \cite[Chap. 1,\S 1.3]{Fulton.Intersection}.
We use $[V]$ to refer to the class of a subvariety $V$ in $A_*(X)$.
\begin{defn}[Pushforward,  {\cite[Chap. 1, p. 11]{Fulton.Intersection}}] 
 Let $f: X\longrightarrow Y$ be a proper morphism. 
Let $V$ be a subvariety of $X$, the image $W=f(V)$ a subvariety of $Y$, and the function field $R(V)$ an extension of the function field $R(W)$. 
The pushforward $f_* : A_*(X)\longrightarrow A_*(Y)$ is defined as follows
\begin{equation}
f_* [V]= \begin{cases}
0  &  \text{if} \quad \dim V\neq \dim W,\\
[R(V):R(W)] \   [V_2]  & \text{if} \quad \dim V= \dim W,
\end{cases}
\end{equation}
where $[R(V):R(W)]$ is the degree of the field extension $R(V)/R(W)$. 
\end{defn}
\begin{defn}[Degree, {\cite[Chap. 1, p. 13]{Fulton.Intersection}}]
The degree of a class $\alpha$ of $A_*(X)$ is denoted by  $\int_X \alpha$ (or simply $\int \alpha$ if there is no ambiguity in the choice of $X$), and is defined to be the degree of its component in $A_0(X)$.
\end{defn}
The total homological Chern class $c(X)$ of any nonsingular variety $X$ of dimension $d$ is defined as
\begin{equation}
c(X)=c(TX)\cap [X],
\end{equation}
where $TX$ is the tangent bundle of $X$ and $[X]$ is the class of $X$ in the Chow ring. The degree of $c(X)$ is the topological Euler characteristic of $X$: 
\begin{equation}
\chi(X)=\int_X c(X).
\end{equation}
The following Lemma gives an important functorial property of the degree.
\begin{lem}[{\cite[Chap. 1, p. 13]{Fulton.Intersection}}]  \label{lem:Push} Let $f:X\longrightarrow Y$  be a proper map between varieties. 
 For any class $\alpha$ in the Chow ring $A_*(X)$ of $X$:
\begin{equation}
\int_X \alpha=\int_Y f_* \alpha.
\end{equation}
\end{lem}
 Lemma \ref{lem:Push} means that an intersection number in $X$ can be computed in $Y$ through a pushforward of a proper map $f:X\longrightarrow Y$. 
 This simple fact  has far-reaching consequences as it allows us to express the topological invariants of an elliptic fibration in terms of those of the base.

Let $X$ be a projective variety with at worst canonical Gorenstein singularities. 
We denote the canonical class by $K_X$. 
\begin{defn}
A birational projective  morphism $\rho:Y\longrightarrow X$ is called a \emph{crepant desingularization} of $X$ if $Y$ is smooth and 
$K_Y=\rho^* K_X$. 
\end{defn}

\begin{defn}
A resolution of singularities of a variety $Y$ is a proper surjective birational morphism $\varphi:\widetilde{Y}\longrightarrow Y$  such that  
$\widetilde{Y}$ is nonsingular
and  $\varphi$ is an isomorphism away  from the singular  locus of $Y$. In other words, $\widetilde{Y}$ is nonsingular and  if $U$ is the singular locus of $Y$, $\varphi$ maps $\varphi^{-1}(Y\setminus U)$  isomorphically  onto $Y\setminus U$.  
 A \emph{crepant resolution of singularities}  is a resolution of singularities such that  $K_Y=f^* K_X$. 
\end{defn}

When pushing forward blowups of  a projective bundle $\pi: X_0=\mathbb{P}[\mathscr{O}_B\oplus\mathscr{L}^{\otimes 2} \oplus \mathscr{L}^{\otimes 3}]\longrightarrow B$, the key ingredients are the following three theorems. 
The first one is a theorem of Aluffi which gives the Chern class after a blowup along a local complete intersection. 
The second theorem is a pushforward theorem that provides a user-friendly method to compute invariant of the blowup space in terms of the original space. 
The last theorem is a direct consequence of functorial properties of the Segre class and gives a simple method to pushforward analytic expressions in the Chow ring of the projective bundle $X_0$ to  the Chow ring of its base. 
We follow mostly \cite{Euler}.

\begin{thm}[Aluffi, {
{\cite[Lemma 1.3]{Aluffi_CBU}}}]
\label{Thm:AluffiCBU}
Let $Z\subset X$ be the  complete intersection  of $d$ nonsingular hypersurfaces $Z_1$, \ldots, $Z_d$ meeting transversally in $X$.  Let  $f: \widetilde{X}\longrightarrow X$ be the blowup of $X$ centered at $Z$. We denote the exceptional divisor of $f$  by $E$. The total Chern class of $\widetilde{X}$ is then:
\begin{equation}
c( T{\widetilde{X}})=(1+E) \left(\prod_{i=1}^d  \frac{1+f^* Z_i-E}{1+ f^* Z_i}\right)  f^* c(TX).
\end{equation}
\end{thm}

\begin{lem}[See {\cite{Aluffi_CBU,Euler}}] \label{Thm:PushE}
\label{lem:symhom}
Let  $f: \widetilde{X}\longrightarrow X$ be the blowup of $X$ centered at $Z$. We denote the exceptional divisor of $f$  by $E$. Then 
\begin{equation}
 f_* E^n=
 (-1)^{d+1} h_{n-d} (Z_1, \ldots, Z_d) Z_1\cdots Z_d,  
\end{equation}
where $h_i(x_1, \ldots, x_k)$ is the complete homogeneous symmetric polynomial of degree $i$ in $(x_1, \ldots, x_k)$ with the convention that $h_i$ is identically zero for $i<0$ and $h_0=1$.  
\end{lem}

\begin{thm}[Esole--Jefferson--Kang,  see  {\cite{Euler}}] \label{Thm:Push}
    Let the nonsingular variety $Z\subset X$ be a complete intersection of $d$ nonsingular hypersurfaces $Z_1$, \ldots, $Z_d$ meeting transversally in $X$. Let $E$ be the class of the exceptional divisor of the blowup $f:\widetilde{X}\longrightarrow X$ centered 
at $Z$.
 Let $\widetilde{Q}(t)=\sum_a f^* Q_a t^a$ be a formal power series with $Q_a\in A_*(X)$.
 We define the associated formal power series  ${Q}(t)=\sum_a Q_a t^a$, whose coefficients pullback to the coefficients of $\widetilde{Q}(t)$. 
 Then the pushforward $f_*\widetilde{Q}(E)$ is
\begin{equation}
  f_*  \widetilde{Q}(E) =  \sum_{\ell=1}^d {Q}(Z_\ell) M_\ell, \quad \text{where} \quad  M_\ell=\prod_{\substack{m=1\\
 m\neq \ell}}^d  \frac{Z_m}{ Z_m-Z_\ell }.
\end{equation}
\end{thm}

\begin{thm}[{\cite{Fullwood:SVW}}]\label{Thm:FW}
Let $\mathscr{L}$ be a line bundle over a variety  $B$. We denote the first Chern class of $\mathscr{L}$by $L=c_1(\mathscr{L})$.  
We define a vector bundle $\mathscr{E}= \mathscr{O}_B\oplus \cdots \oplus \mathscr{L}^{r_{n}}$ of rank $n+1$  
and the projective bundle of lines $\pi:X_0=\mathbb{P}(\mathscr{E})\longrightarrow B$. 
We denote by $\mathscr{O}_{X_0}(1)$ the dual of the tautological line of $\mathbb{P}(\mathscr{E})$. 
The first Chern class of $\mathscr{O}_{X_0}(1)$ is $c_1(\mathscr{O}_{X_0}(1))=H$. 
We assume that the rational fraction expansion of $\frac{1}
{c(\mathscr{E})}=\prod_{i=1}^{n} \frac{1}{(1+r_i L)}$ is 
\begin{equation}
\frac{1}
{c(\mathscr{E})}
=\sum_{i=1}^m\sum_{j=1}^{n_i} q_{ij} \frac{1}{(1+r_i L)^j}.
\end{equation}
Let $F=\pi^*F_0+ \pi^*F_1 H+ \pi^*F_2 H^2+\cdots$ be an analytic function of $H$ with  $F_i$ in the Chow ring of $B$. 
Then 
\begin{equation}
\pi_*(F)=\sum_{i=1}^m\sum_{j=1}^{n_i} q_{ij} \frac{1}{(j-1)!}\frac{d^{j-1}}{dH^{j-1}}\Big[\frac{1}{H^{n-j}}(F-F_0 -F_1 H -\cdots -F_{n-1}H^{n-1})\Big]\Big|_{H=-r_i L}.
\end{equation}
\end{thm}

\begin{thm}[\cite{Euler}]\label{Thm:PushH}
Let $\mathscr{L}$ be a line bundle over a variety $B$ and $\pi: X_0=\mathbb{P}(\mathscr{O}_B\oplus\mathscr{L}^{\otimes 2} \oplus \mathscr{L}^{\otimes 3})\longrightarrow B$ a projective bundle over $B$. 
 Let $\widetilde{Q}(t)=\sum_a \pi^* Q_a t^a$ be a formal power series in  $t$ such that $Q_a\in A_*(B)$. Define the auxiliary power series $Q(t)=\sum_a Q_a t^a$. 
Then 
\begin{equation}
\pi_* \widetilde{Q}(H)=-2\left. \frac{{Q}(H)}{H^2}\right|_{H=-2L}+3\left. \frac{{Q}(H)}{H^2}\right|_{H=-3L}  +\frac{Q(0)}{6 L^2},
\end{equation}
 where  $L=c_1(\mathscr{L})$ and $H=c_1(\mathscr{O}_{X_0}(1))$ is the first Chern class of the dual of the tautological line bundle of  $ \pi:X_0=\mathbb{P}(\mathscr{O}_B \oplus\mathscr{L}^{\otimes 2} \oplus\mathscr{L}^{\otimes 3})\rightarrow B$.
\end{thm}
In particular, we have
\begin{equation}
\begin{aligned}
\begin{cases}
\pi_* 1 = 0, \quad \pi_*H= 0, \quad \pi_*H^2= 1 , \quad  \pi_* H^3 = -5 L, \quad\pi_*  H^4= 19 L^2, \quad \pi_* H^5= -65 L^3,\\
\pi_*( H^{k}) =\left[(-2)^{k-1} -(-3)^{k-1} \right] L^{k-2} \quad n\geq 1.
\end{cases}
\end{aligned}
\end{equation}

\begin{thm}\label{lem:PushH1}
Given the projective bundle $ \pi:X_0=
\mathbb{P}(\mathscr{O}_B\oplus\mathscr{L}\oplus \mathscr{L}^{\otimes 2}\oplus \mathscr{L}^{\otimes 2})\rightarrow B$, we 
denoting the first Chern class of $\mathscr{L}$ by $L$ and the first Chern class of the dual of the tautological line bundle of $\pi$ by $H$. We have

\begin{equation}
\begin{aligned}
\pi_* F(H) &=\Big( \frac{F-F_0-F_1 H-F_2 H^2}{H^3}\Big)\Big|_{H=-L}-2\Big(\frac{F-F_0-F_1 H-F_2 H^2}{H^3}\Big)\Big|_{H=-2L} \\
& \  \  +
\partial_H\Big(\frac{F-F_0-F_1 H-F_2 H^2}{H^2}\Big) \Big|_{H=-2L}  .
\end{aligned}
\end{equation}
\end{thm}
In particular, 
\begin{equation}
\begin{cases}
\pi_* 1=0, \quad  \pi_* H =0, \quad \pi_* H=0, \quad \pi_* H^2=0, \quad \pi_* H^3=1,\\
\pi_* H^4=- 5 L, \quad \pi_* H^5= 17 L^2, \quad \pi_* H^6=- 49 L^3, \quad \pi_* H^7=129 L^4, \quad H^7= - 321 L^5, 
\\
 H^{3+k}= (-L)^k -2 (-2L)^k -(k+1) (-2L)^{k}.
 \end{cases}
\end{equation}
\begin{proof}
We use Theorem \ref{Thm:FW} with the  partial fraction decomposition :
\begin{equation}
\frac{1}{(1+L)(1+2L)^2}=\frac{1}{1+L}-\frac{2}{1+2 L}+\frac{2}{(1+2 L)^2}.\nonumber
\end{equation}
The first few terms and the generic term can be computed directly as follows
\begin{equation}
\begin{aligned}\nonumber
\frac{1}{(1+L)(1+2L)^2}&=\frac{1}{1+L}-\frac{2}{1+2 L}-\partial_L\Big(\frac{1}{1+2 L}\Big) \\
&= \sum_{k=0}^\infty\Big( (-L)^k  -2 (-2L)^k     - 2(k+1) (-2L)^k\Big)\\
&=1 - 5 L + 17 L^2 - 49 L^3 + 129 L^4 - 321 L^5 + 769 L^6 +\cdots .
\end{aligned}
\end{equation}
\end{proof}

\begin{thm}\label{lem:PushH2}
Given the projective bundle $ \pi:X_0=\mathbb{P}(\mathscr{O}_B\oplus\mathscr{L}\oplus \mathscr{L}^{\otimes 2})
\rightarrow B$.
Denoting the first Chern class of $\mathscr{L}$ by $L$ and the first Chern class of the dual of the tautological line bundle of $\pi$ by $H$, we have:
\begin{equation}
\begin{aligned}
\pi_* F(H) &=-\Big( \frac{F-F_0-F_1 H}{H^2}\Big)\Big|_{H=-L}+2\Big(\frac{F-F_0-F_1 H}{H^2}\Big)\Big|_{H=-2L}  .
\end{aligned}
\end{equation}
\end{thm}
In particular,
\begin{equation}
\begin{cases}
\pi_*  1=0, \quad \pi_* H =0, \quad  \pi_* H^2=1, \\
\pi_* H^3= - 3 L, \quad \pi_* H^4= + 7 L^2, \quad \pi_* H^5 =- 15 L^3, \quad \pi_* H^6= + 31 L^4, \quad \pi_* H^7= - 63 L^5,
\\
 \pi_* H^{k+2}= -(-L)^k  +2 (-2L)^k \quad (k>0).
 \end{cases}
\end{equation}

\begin{proof}
We use Theorem \ref{Thm:FW} with the  partial fraction decomposition:
$$
\begin{aligned}
\frac{1}{(1+L)(1+2L)}&=-\frac{1}{1+L}+\frac{2}{1+2 L}\\
&= \sum_{k=0}^\infty\Big( -(-L)^k  +2 (-2L)^k  \Big)\\
&=1 - 3 L + 7 L^2 - 15 L^3 + 31 L^4 - 63 L^5 + 127 L^6 +\cdots .
\end{aligned}
$$
\end{proof}

\begin{thm}\label{lem:PushH3}
Given the projective bundle $ \pi:X_0=\mathbb{P}(\mathscr{O}_B\oplus\mathscr{L}\oplus \mathscr{L}) \rightarrow B$. Denoting the first Chern class of $\mathscr{L}$ by $L$ and the first Chern class of the dual of the tautological line bundle of $\pi$ by $H$, we have:
\begin{equation}
\pi_* F=\Big( \partial_L \frac{F-F_0-F_1 H}{H^2}\Big) \Big|_{H=-L} .
\end{equation}
\end{thm}
In particular, we have 
\begin{equation}
\begin{cases}
\pi_*  1=0, \quad \pi_* H =0, \quad  \pi_* H^2=1, \\
\pi_* H^3=- 2 L, \quad \pi_* H^4=+ 3 L^2, \quad \pi_* H^5= - 4 L^3, \quad \pi_* H^6=+ 5 L^4, \quad \pi_* H^7= - 6 L^5, \\
 \pi_* H^{k+2}= (k+1) (-L)^k \quad (k>0).
 \end{cases}
\end{equation}
\begin{proof}
The result for $\pi_* F$ follows directly from  Theorem \ref{Thm:FW}. 
We can compute the first few terms by the following expansion:
$$
\begin{aligned}
\frac{1}{(1+L)^2}&=-\partial_L\frac{1}{1+L}=-\partial_L\Big(\sum_{k=0}^\infty (-L)^k\Big)\\
&=1 - 2 L + 3 L^2 - 4 L^3 + 5 L^4 - 6 L^5 + 7 L^6 +\cdots .
\end{aligned}
$$
\end{proof}

\begin{thm}\label{lem:PushH4}
Given the projective bundle $ \pi:X_0=\mathbb{P}(\mathscr{O}_B\oplus\mathscr{L}\oplus \mathscr{L}\oplus\mathscr{L})
\rightarrow B$.
Denoting the first Chern class of $\mathscr{L}$ by $L$ and the first Chern class of the dual of the tautological line bundle of $\pi$ by $H$, we have:
\begin{equation}
\pi_* F=\frac{1}{2}\Big( \partial^2_L \frac{F-F_0-F_1 H-F_2 H^2}{H}\Big)_{H=-L}.
\end{equation}
\end{thm}
In particular, we have, 
\begin{equation}
\begin{cases}
\pi_*  1=0, \quad \pi_* H =0, \quad  \pi_* H^2=0, \quad \pi_* H^3=1\\
\pi_* H^4=-3 L,\quad \pi_* H^4=+6 L^2,\quad \pi_* H^5=-10 L^3, \quad \pi_* H^6=15 L^4, \quad \pi_* H^7=-21 L^5,\\
 \pi_* H^{k+3}= \frac{1}{2}(k+1)(k+2) (-L)^k \quad (k>0).
 \end{cases}
\end{equation}\begin{proof}
The result for $\pi_* F$ follows directly from  Theorem \ref{Thm:FW}. 
We can compute the first few terms by the following expansion:
$$
\begin{aligned}
\frac{1}{(1+L)^3}&=\frac{1}{2}\partial^2_L\frac{1}{(1+L)}=\frac{1}{2}\partial_L^2  \sum_{k=0}^\infty (-L)^k\\
& =1-3 L+6 L^2-10 L^3+15 L^4-21 L^5+28 L^6+\cdots .
\end{aligned}
$$
\end{proof}

\section{Collection of results} \label{sec:results}

The following theorem gives the behaviors of intersection numbers involving Chern classes and Pontryagin classes of dimension too small to give Chern or Pontryagin numbers \cite{Chern}. 
To give a number, they must be multiplied by an element of the Chow ring of appropriate dimension.

\begin{thm}\label{thm.1.4}
Let  $\varphi:Y\longrightarrow B$ be an  elliptic fibration
given by the crepant resolution of a singular Weierstrass model of dimension $n$. Then,
\begin{align}
& \int_Y c_1^i(TY)\cdot  \alpha &&=\int_B (c_1-L)^i\cdot \varphi_* \alpha, \quad &&\alpha \in A^*(Y),\label{Thm21.1}  \\
& \int_Y c_1^i(TY)\cdot  \varphi^* \beta &&=\int_B (c_1-L)^i \beta,  &&\beta \in A^*(B), \label{Thm21.2} \\
& \int_Y c_1^n(TY) &&=\int_B (c_1-L)^n=0, \label{Thm21.3}\\
& \int_Y c_3(TY)\cdot  \varphi^* \beta  &&=\int_B \varphi_* (c_3 (TY) [Y]) \cdot \beta, \quad &&\beta \in A^*(B), \label{Thm21.4}\\
& \int_Y c_2(TY)\cdot  \varphi^* \beta  &&=12\int_B L \cdot \beta, \quad &&\beta \in A^*(B), \label{Thm21.5}\\
& \int_Y p_1(TY)\cdot  \varphi^* \beta  &&=
 \int_B(c_1-L)^2 \beta-24\int_B L\cdot \beta, \quad &&\beta \in A^*(B),\label{Thm21.6}
\end{align}
where $ \varphi_* (c_3 (TY) [Y])$ are given in Table \ref{Table:ChernNumbers3} and Table \ref{Table:ChernNumbersG}.
\end{thm}

\begin{proof}
Since $Y$ is a crepant resolution of a Weierstrass model, we have that the first Chern class of $Y$ is the pullback $c_1(TY)=\varphi^*(c_1-L)$.
Equation \eqref{Thm21.1} is therefore a direct consequence of the projection formula and the invariance of the degree under a proper map: 
\begin{equation}
\int_Y \varphi^* (c_1-L)^i \alpha =\int_B \varphi_* (\varphi^* (c_1-L)^i \alpha)= \int_B (c_1-L)^i \varphi_*\alpha.
\end{equation}
Equations \eqref{Thm21.2} and \eqref{Thm21.3} are direct specializations of equation \eqref{Thm21.1}. 
In particular, if $Y$ is an $n$-fold, we have $\int_Y \varphi^* (c_1-L)^n=\int_B (c_1-L)^n=0$. 
 Equation \eqref{Thm21.4} is also a direct consequence of the projection formula.
 Equation \eqref{Thm21.5} follows from direct computation of the pushforward of $c_2$ for each models. The answer is always $12L$ as in the case of a smooth Weierstrass model. 
 Equation \eqref{Thm21.6}  is derived by linear combination using $p_1=c_1^2-2c_2$ and equations \eqref{Thm21.2} and \eqref{Thm21.5}.
  \end{proof}

  \begin{cor}\label{Cor.c_2}
  Let  $\varphi:Y\longrightarrow B$ be an  elliptic fibration
given by the crepant resolution of a singular Weierstrass model of dimension $n$ with fundamental line bundle $\mathscr{L}$. Then,
\begin{equation}
  \int_Y c_1(TY)^{n-1}c_2(TY) = 12\int_B (c_1-L)^{n-1} L.
 \end{equation}
  \end{cor}

Below are tables summarizing all the results. 

The E$_8$, E$_7$, E$_6$, D$_5$, and Q$_7$-models are considered in Tables \ref{Table:ChernNumbers3}--\ref{Table.Pontryagin2}. The pushforwards of the third Chern classes are given in Table \ref{Table:ChernNumbers3}, the Chern numbers of the elliptically fibered fourfolds are given in Table \ref{Table.ChernNumbers}, the holomorphic genera of the fourfolds are given in Table \ref{Table.HolomorphicEC}, the Pontryagin numbers are given in Table \ref{Table.Pontryagin}, and the Hirzebruch signatures $\sigma$, the A-genus $\hat{A}_2$, and the curvature invariants $X_8$ are given in Table \ref{Table.Pontryagin2}.

The $G$-models for $G=\text{SO}(n)$ for $n=3,4,5,6$ and $G=\text{PSU}(3)$ are considered in Tables \ref{tab:blowupcenters}--\ref{Table.Pontryagin2G}. The sequence of blowups used to describe a crepant resolution for each $G$-model is summarized in Table \ref{tab:blowupcenters}, the pushforwards of the third Chern classes are given in Table \ref{Table:ChernNumbersG}, the Chern numbers of the elliptically fibered fourfolds are given in Table \ref{Table.ChernNumbersG}, the holomorphic genera of the fourfolds are given in Table \ref{Table.HolomorphicECG}, the Pontryagin numbers are given in Table \ref{Table.PontryaginG}, and the Hirzebruch signatures $\sigma$, the A-genus $\hat{A}_2$, and the curvature invariants $X_8$ are given in Table \ref{Table.Pontryagin2G}.

\begin{table}[hb!]
\begin{center}
\renewcommand{\arraystretch}{2}
\scalebox{1}{ $
\begin{array}{|c|c|}
\hline 
\text{Type} & \varphi_* \Big(c_3(TY)[Y]\Big) \\
\hline
\text{Q}_7 & 6 \left(2 c_1 L-6 L^2+L S-S^2\right) \\
\hline
\text{D}_5 & 4 L (3 c_1 - 7 L) \\
\hline 
\text{E}_6  & 12 L (c_1 - 3 L) \\
\hline 
\text{E}_7 & 12 L (c_1 - 4 L) \\
\hline 
\text{E}_8 & 12 L(c_1 - 6 L) \\
\hline 
\end{array} 
$ }
\end{center}
\caption{Chern numbers after pushforwards to the base. The divisor $S$ appears in the definition of the $Q_7$-model,  $L=c_1(\mathscr{L})$, and  $c_i$ denotes the $i$th Chern class of the base of the fibration.}
 \label{Table:ChernNumbers3}
\end{table}

\begin{table}[htb!]
\begin{center}
\renewcommand{\arraystretch}{2}
\scalebox{.8}{$
\begin{array}{|c|c|c|c|c|c|}
\hline 
\text{Type} & c_1^4 (TY) & c_1^2 (TY) c_2(TY) & c_2^2 (TY) & c_1(TY) c_3(TY) & c_4(TY) \\\hline
\text{Q}_7 & 0 & 12 L (c_1-L)^2 & \begin{array} {c}  
2 (12 c_2 L - 12 c_1 L^2 + 22 L^3)\\
-2 S (5 L^2 - 8 L S + S^2)
\end{array} & \begin{array} {c} 
6 L (c_1-L) (2 c_1-6 L) \\
+6 S (c_1-L) (L - S)
\end{array} & \begin{array} {c} 
6 L (2 c_2 - 6 c_1 L+ 16 L^2) \\
+6 L S (c_1-6 L) \\
-6 S^2 (c_1-9 L) - 6 S^3
\end{array} \\
\hline 
\text{D}_5 & 0 & 12 L (c_1-L)^2 & \begin{array} {c}  
12 L (2 c_2 - 2 c_1 L + 3 L^2)
\end{array} & \begin{array} {c} 
4 L (c_1 - L)(3 c_1 - 7 L)
\end{array} & \begin{array} {c} 
4 L (3 c_2 - 7 c_1 L + 16 L^2)
\end{array} \\
\hline 
\text{E}_6 & 0 & 12 L (c_1-L)^2 & \begin{array} {c} 
24 L (c_2 - c_1 L + 2 L^2)
\end{array} & \begin{array} {c} 
12 L (c_1-L) (c_1-3 L)
\end{array} & \begin{array} {c} 
12 L (c_2 - 3 c_1 L + 9 L^2)
\end{array} \\
\hline 
\text{E}_7 & 0 & 12 L (c_1-L)^2 & \begin{array} {c}  
24 L (c_2 - c_1 L + 3 L^2)
\end{array} & \begin{array} {c} 
12 L (c_1-L) (c_1-4 L)
\end{array} & \begin{array} {c} 
12 L (c_2 - 4 c_1 L + 16 L^2)
\end{array} \\
\hline 
\text{E}_8 & 0 & 12 L (c_1-L)^2 & \begin{array} {c}  
24 L (c_2 - c_1 L + 6 L^2)
\end{array} & \begin{array} {c} 
12 L (c_1 - 6 L) (c_1 - L)
\end{array} & \begin{array} {c} 
12 L (c_2 - 6 c_1 L + 36 L^2)
\end{array} \\
\hline 
\end{array}
$}
\end{center}
\caption{Chern numbers of elliptically fibered fourfolds obtained from crepant resolutions of Tate's models. 
We abuse notation and omit the degree $\int$ in the entries of the table. 
The divisor $S$ appears in the definition of the $Q_7$-model,  $L=c_1(\mathscr{L})$, and  $c_i$ denotes the $i$th Chern class of the base of the fibration. 
}
\label{Table.ChernNumbers}
\end{table}

\begin{table}[htb!]
\begin{center}
\renewcommand{\arraystretch}{2}
\scalebox{.85}{$
\begin{array}{|c|c|c|c|}
\hline 
\text{Type} & \chi_0 & \chi_1 & \chi_2 \\\hline
\text{Q}_7 & \begin{array} {c} 
\frac{1}{12} L \left(c_1^2-3 c_1 L+c_2+2 L^2\right)
\end{array} & \begin{array} {c} 
-\frac{1}{3} L \left(2 c_1^2-27 c_1 L+5 c_2+55 L^2\right)\\
-\frac{1}{2} L S (3 c_1-13 L)\\
+\frac{1}{2} S^2 (3 c_1-19 L)+S^3
\end{array} & \begin{array} {c} 
-\frac{1}{2} L \left(3 c_1^2+35 c_1 L-17 c_2-118 L^2\right)\\
+L S (3 c_1-23 L)\\
 -3S^2  (c_1-35 L)-4S^3\\
\end{array} \\ 
\hline 
\text{D}_5 & \begin{array} {c} 
\frac{1}{12} L \left(c_1^2-3 c_1 L+c_2+2 L^2\right)
\end{array} & \begin{array} {c} 
-\frac{1}{3} L \left(2 c_1^2-21 c_1 L+5 c_2+37 L^2\right)
\end{array} & \begin{array} {c} 
\frac{1}{2} L \left(-3 c_1^2-27 c_1 L+17 c_2+78 L^2\right)
\end{array} \\ 
\hline 
\text{E}_6 & \begin{array} {c} 
\frac{1}{12} L \left(c_1^2-3 c_1 L+c_2+2 L^2\right)
\end{array} & \begin{array} {c} 
-\frac{1}{3} L \left(2 c_1^2-27 c_1 L+5 c_2+61 L^2\right)
\end{array} & \begin{array} {c} 
\frac{1}{2} L \left(-3 c_1^2-35 c_1 L+17 c_2+134 L^2\right)
\end{array} \\  
\hline 
\text{E}_7 & \begin{array} {c} 
\frac{1}{12} L \left(c_1^2-3 c_1 L+c_2+2 L^2\right)
\end{array} & \begin{array} {c} 
-\frac{1}{3} L \left(2 c_1^2-36 c_1 L+5 c_2+106 L^2\right)
\end{array} & \begin{array} {c} 
\frac{1}{2} L \left(-3 c_1^2-47 c_1 L+17 c_2+242 L^2\right)
\end{array} \\ 
\hline 
\text{E}_8 & \begin{array} {c} 
\frac{1}{12} L \left(c_1^2-3 c_1 L+c_2+2 L^2\right)
\end{array} & \begin{array} {c} 
-\frac{1}{3} L \left(2 c_1^2-54 c_1 L+5 c_2+232 L^2\right)
\end{array} & \begin{array} {c} 
\frac{1}{2} L \left(-3 c_1^2-71 c_1 L+17 c_2+554 L^2\right)
\end{array} \\ 
\hline 
\end{array}
$}
\end{center}
\caption{Holomorphic genera.  To ease the notation, we  omit the degree $\int$ in the entries of the table. 
The divisor $S$ appears in the definition of the $Q_7$-model,  $L=c_1(\mathscr{L})$, and  $c_i$ denotes the $i$th Chern class of the base of the fibration.
The holomorphic Euler characteristic $\chi_0(Y)$ is equal to $\chi_0(W,\mathscr{O}_W)$ where $W$ is the divisor defined by $\mathscr{L}$ in 
the base.}
\label{Table.HolomorphicEC}
\end{table}

\begin{table}[htb!]
\begin{center}
\renewcommand{\arraystretch}{2}
\scalebox{1}{$
\begin{array}{|c|l|l|}
\hline 
\text{Type}  & \hspace{1cm}\int_Y p_2(TY) & \hspace{1cm} \int_Y p_1^2(TY)  \\\hline
\text{Q}_7 & \begin{array}{c} 4 L \left(-6 c_1^2+12 c_2+41 L^2\right) \\ -14 S \left(5 L^2-8 L S+S^2\right) \end{array}
 & \begin{array}{c} 16 L (-3 c_1^2+6 c_2+8 L^2)\\ -8 S (5 L^2 - 8 L S + S^2) \end{array}
  \\
\hline 
\text{D}_5 & 12 L \left(-2 c_1^2+4 c_2+9 L^2\right)
 & 48 L \left(-c_1^2+2 c_2+2 L^2\right)
  \\
\hline 
\text{E}_6 & 24 L \left(-c_1^2+2 c_2+8 L^2\right)
 & 48 L \left(-c_1^2+2 c_2+3 L^2\right)
 \\
\hline 
\text{E}_7 & 24 L \left(-c_1^2+2 c_2+15 L^2\right)
 & 48 L \left(-c_1^2+2 c_2+5 L^2\right)
  \\
\hline 
\text{E}_8 & 24 L (-c_1^2+2 c_2+36 L^2)
 & 48 L (-c_1^2+2 c_2+11 L^2)
 \\
\hline 
\end{array}
$}
\end{center}
\caption{Pontryagin numbers. To ease the notation, we omit the degree $\int$ in the entries of the table. 
The divisor $S$ appears in the definition of the $Q_7$-model,  $L=c_1(\mathscr{L})$, and  $c_i$ denotes the $i$th Chern class of the base of the fibration. 
} 
\label{Table.Pontryagin}
\end{table}

\begin{table}[htb!]
\begin{center}
\renewcommand{\arraystretch}{1.6}
\scalebox{.9}{$
\begin{array}{|c|c|c|c|c|c|c|}
\hline 
\text{Type} & 192 X_8=\int_Y( p_1^2-4p_2) &45 \sigma=45\int_Y L_2= \int_Y(7p_2-p_1^2) & 5760\int_Y \hat{\text{A}}_2=\int_Y(7 p_1^2-4 p_2) \\\hline
\text{Q}_7 & \begin{array} {c} 
48 (c_1^2 L - 2 c_2 L - 11 L^3)\\
+48 S (5 L^2 - 8 L S + S^2)
\end{array} & \begin{array} {c} 
60 L (-2 c_1^2+4 c_2+17 L^2)\\
-90 S (5 L^2 - 8 L S + S^2)
\end{array} & \begin{array} {c} 
240 L \left(-c_1^2+2 c_2+L^2\right)
\end{array} \\
\hline 
\text{D}_5 & \begin{array} {c} 
48 L \left(c_1^2-2 c_2-7 L^2\right)
\end{array} & \begin{array} {c} 
60 L \left(-2 c_1^2+4 c_2+11 L^2\right)
\end{array} & \begin{array} {c} 
240 L \left(-c_1^2+2 c_2+L^2\right)
\end{array} \\
\hline 
\text{E}_6 & \begin{array} {c} 
48 L \left(c_1^2-2 c_2-13 L^2\right)
\end{array} & \begin{array} {c} 
120 L \left(-c_1^2+2 c_2+10 L^2\right)
\end{array} & \begin{array} {c} 
240 L \left(-c_1^2+2 c_2+L^2\right)
\end{array} \\
\hline 
\text{E}_7 & \begin{array} {c} 
48 L \left(c_1^2-2 c_2-25 L^2\right)
\end{array} & \begin{array} {c} 
120 L \left(-c_1^2+2 c_2+19 L^2\right)
\end{array} & \begin{array} {c} 
240 L \left(-c_1^2+2 c_2+L^2\right)
\end{array} \\
\hline 
\text{E}_8 & \begin{array} {c} 
48 L \left(c_1^2-2 c_2-61 L^2\right)
\end{array} & \begin{array} {c} 
120 L \left(-c_1^2+2 c_2+46 L^2\right)
\end{array} & \begin{array} {c} 
240 L \left(-c_1^2+2 c_2+L^2\right)
\end{array} \\
\hline 
\end{array}
$}
\end{center}
\caption{Characteristic invariants: the anomaly invariant $X_8$, the signature $\sigma$, and the index of the $\hat{\text{A}}$-genus. 
To ease the notation, we abuse notation and omit the degree $\int$ in the entries of the table. 
The divisor $S$ appears in the definition of the $Q_7$-model,  $L=c_1(\mathscr{L})$, and  $c_i$ denotes the $i$th Chern class of the base of the fibration.
}
\label{Table.Pontryagin2}
\end{table}
\begin{table}[h!]
\centering
\scalebox{1}{$
\begin{array}{|c|c|c|}
	\hline
	\text{Group} & \text{Fiber Type} & \text{Crepant Resolution} \\\hline
 \text{SO}(3) & \text{I}_2^\text{ns},\   \text{III} & \begin{array}{c} \begin{tikzpicture}
	\node(X0) at (0,0){$X_0$};
	\node(X1) at (2.5,0){$X_1$};
	\draw[big arrow] (X1) -- node[above,midway]{$(x,y|e_1)$} (X0);		
 \end{tikzpicture}\end{array} \\\hline
 \text{PSU}(3) & \text{I}_3^{\text{s}} & \begin{array}{c} \begin{tikzpicture}
	\node(X0) at (0,0){$X_0$};
	\node(X1) at (2.5,0){$X_1$};
	\node(X2) at (5,0){$X_2$};
	\draw[big arrow] (X1) -- node[above,midway]{$(x,y,s|e_1)$} (X0);	
	\draw[big arrow] (X2) -- node[above,midway]{$(y,e_1|e_2)$} (X1);		
 \end{tikzpicture}\end{array}  \\\hline
 \text{SO}(4) & \begin{array}{c} \text{I}_2^{\text{ns}}+\text{I}_2^{\text{s}} \\ \text{I}_2^{\text{s}}+\text{I}_2^{\text{s}} \\ \text{III}+\text{I}_2^{\text{ns}} \\  \text{III}+\text{I}_2^{\text{s}} \\  \text{III}+\text{III} \end{array} & \begin{array}{c} \begin{tikzpicture}
	\node(X0) at (0,0){$X_0$};
	\node(X1) at (2.5,0){$X_1$};
	\node(X2) at (5,0){$X_2$};
	\draw[big arrow] (X1) -- node[above,midway]{$(x,y,s|e_1)$} (X0);	
	\draw[big arrow] (X2) -- node[above,midway]{$(x,y,t|w_1)$} (X1);		
 \end{tikzpicture}\end{array}  \\\hline
 \text{SO}(5) & \text{I}_4^{\text{ns}} & \begin{array}{c} \begin{tikzpicture}
	\node(X0) at (0,0){$X_0$};
	\node(X1) at (2.5,0){$X_1$};
	\node(X2) at (5,0){$X_2$};
	\draw[big arrow] (X1) -- node[above,midway]{$(x,y,s|e_1)$} (X0);	
	\draw[big arrow] (X2) -- node[above,midway]{$(x,y,e_1|e_2)$} (X1);		
 \end{tikzpicture}\end{array}  \\\hline
 \text{SO}(6) & \text{I}_4^\text{s} & \begin{array}{c} 
 \begin{tikzpicture}
	\node(X0) at (0,0){$X_0$};
	\node(X1) at (2.5,0){$X_1$};
	\node(X2) at (5,0){$X_2$};
	\node(X3) at (7.5,0){$X_3$};
	\draw[big arrow] (X1) -- node[above,midway]{$(x,y,s|e_1)$} (X0);	
	\draw[big arrow] (X2) -- node[above,midway]{$(y,e_1|e_2)$} (X1);
	\draw[big arrow] (X3) -- node[above,midway]{$(x,e_2|e_3)$} (X2);		
 \end{tikzpicture}\end{array} \\\hline
	\end{array}
$}
\caption{The blowup centers of the crepant resolutions. The variable $s$ and $t$ is a section of the line bundles $\mathscr{O}_B(S)$ and $\mathscr{O}_B(T)$.
The SO($4$)-model has reducible Kodaira fibers supported on smooth divisors of classes $T$ and $S=4L-T$. For the notation, see Section \ref{PSU(3)model}. }
\label{tab:blowupcenters}
\end{table}

\begin{table}[htb!]
\begin{center}
\renewcommand{\arraystretch}{1.8}
\scalebox{1}{ $
\begin{array}{|c|c|c|c|c|}
\hline 
\text{Algebra} & \text{Mordell--Weil} & \text{G} & \text{Kodaira} & \varphi_* \Big(c_3(TY)[Y]\Big) \\
\hline
\text{A}_1 & \mathbb{Z}_2 & \text{SO}(3) & \text{I}_2^{\text{ns}}, \text{III} & 12 L (c_1 - 4 L) \\
\hline
\text{A}_2 & \mathbb{Z}_3 & \text{PSU}(3) & \text{I}_3^{\text{s}} & 12 L (c_1 - 3 L) \\
\hline 
\text{A}_1+\text{A}_1 & \mathbb{Z}_2 & \text{SO}(4) & \begin{array}{c} \text{I}_2^{\text{ns}}+\text{I}_2^{\text{s}} \\ \text{I}_2^{\text{s}}+\text{I}_2^{\text{s}} \\ \text{III}+\text{I}_2^{\text{ns}} \\  \text{III}+\text{I}_2^{\text{s}} \\  \text{III}+\text{III} \end{array} & 12 L (c_1 - 4 L) + 16 L T - 4 T^2 \\
\hline 
\text{B}_2 & \mathbb{Z}_2 & \text{SO}(5) & \text{I}_4^{\text{ns}} & 4L (3 c_1 - 8 L) \\
\hline 
\text{A}_3 & \mathbb{Z}_2 & \text{SO}(6) & \text{I}_4^{\text{s}} & 12 L (c_1 - 2 L) \\

\hline 
\end{array} 
$ }
\end{center}
\caption{Chern numbers after pushforwards to the base. 
The SO($4$)-model has reducible Kodaira fibers supported on smooth divisors of classes $T$ and $S=4L-T$. 
By definition, $L=c_1(\mathscr{L})$ and  $c_i$ denotes the $i$th Chern class of the base of the fibration.}
 \label{Table:ChernNumbersG}
\end{table}

\begin{table}[htb!]
\begin{center}
\renewcommand{\arraystretch}{2}
\scalebox{.8}{$
\begin{array}{|c|c|c|c|c|c|}
\hline 
\text{Type} & c_1^4 (TY) & c_1^2 (TY) c_2(TY) & c_2^2 (TY) & c_1(TY) c_3(TY) & c_4(TY) \\\hline
\text{SO}(3) & 0 & 12 L (c_1-L)^2 & \begin{array} {c}  
24 L (c_2 -c_1 L + 3 L^2)
\end{array} & \begin{array} {c} 
12 L (c_1 - L) (c_1 - 4 L)
\end{array} & \begin{array} {c} 
12 L (c_2 -4 c_1 L + 16 L^2)
\end{array} \\
\hline 
\text{PSU}(3) & 0 & 12 L (c_1-L)^2 & \begin{array} {c}  
24 L (c_2 -c_1 L + 2 L^2)
\end{array} & \begin{array} {c} 
12 L (c_1 - 3 L) (c_1 - L)
\end{array} & \begin{array} {c} 
12 L (c_2 - 3 c_1 L + 9 L^2)
\end{array} \\
\hline 
\text{SO}(4) & 0 & 12 L (c_1-L)^2 & \begin{array} {c}  
24 L (c_2 - c_1 L + 3 L^2)\\
- 32 L^2 T + 8 L T^2
\end{array} & \begin{array} {c} 
12 L (c_1 - L) (c_1 - 4 L)\\
+16 T (c_1 - L) (L - 4 T)
\end{array} & \begin{array} {c} 
12 L (c_2 - 4 c_1 L + 16 L^2)\\
-4T (c_1 - 7 L) (T - 4L)
\end{array} \\
\hline 
\text{SO}(5) & 0 & 12 L (c_1-L)^2 & \begin{array} {c}  
8 L (3 c_2 -3 c_1L + 5 L^2)
\end{array} & \begin{array} {c} 
4 L (c_1 - L) (3 c_1 - 8 L)
\end{array} & \begin{array} {c} 
4 L (3 c_2 -8 c_1 L + 20 L^2)
\end{array} \\
\hline 
\text{SO}(6) & 0 & 12 L (c_1-L)^2 & \begin{array} {c}
8 L (3 c_2+4 L^2-3 c_1 L)
\end{array} & \begin{array} {c} 
12 L (c_1-L) (c_1-2 L)
\end{array} & \begin{array} {c} 
12 L (c_2-2 c_1+4 L^2)
\end{array} \\
\hline 
\end{array}
$}
\end{center}
\caption{Chern numbers of elliptically fibered fourfolds obtained from crepant resolutions of Tate's models. 
We abuse notation and omit the degree $\int$ in the entries of the table. 
The SO($4$)-model has reducible Kodaira fibers supported on smooth divisors of classes $T$ and $S=4L-T$. 
By definition, $L=c_1(\mathscr{L})$ and  $c_i$ denotes the $i$th Chern class of the base of the fibration.}
\label{Table.ChernNumbersG}
\end{table}

\begin{table}[htb!]
\begin{center}
\renewcommand{\arraystretch}{2}
\scalebox{.85}{$
\begin{array}{|c|c|c|c|}
\hline 
\text{Type} & \chi_0 & \chi_1 & \chi_2 \\\hline
\text{SO}(3) & \begin{array} {c} 
\frac{1}{12} L (c_1^2 + c_2 - 3 c_1 L + 2 L^2)
\end{array} & \begin{array} {c} 
-\frac{1}{3} L (2 c_1^2 + 5 c_2 - 36 c_1 L + 106 L^2)
\end{array} & \begin{array} {c} 
\frac{1}{2} L (-3 c_1^2 + 17 c_2 - 47 c_1 L + 242 L^2)
\end{array} \\ 
\hline 
\text{PSU}(3) & \begin{array} {c} 
\frac{1}{12} L (c_1^2 + c_2 - 3 c_1 L + 2 L^2)
\end{array} & \begin{array} {c} 
-\frac{1}{3} L (2 c_1^2 + 5 c_2 - 27 c_1 L + 61 L^2)
\end{array} & \begin{array} {c} 
\frac{1}{2} L (-3 c_1^2 + 17 c_2 - 35 c_1 L + 134 L^2)
\end{array} \\ 
\hline 
\text{SO}(4) & \begin{array} {c} 
\frac{1}{12} L (c_1^2 + c_2 - 3 c_1 L + 2 L^2)
\end{array} & \begin{array} {c} 
-\frac{1}{3} L (2 c_1^2 + 5 c_2 - 36 c_1 L + 106 L^2) \\
-T (c_1 - 5 L) (4L-T)
\end{array} & \begin{array} {c} 
-\frac{1}{2} L (3 c_1^2 - 17 c_2 + 47 c_1 L - 242 L^2)\\
+2 T (c_1 - 9 L) (4 L - T)
\end{array} \\ 
\hline 
\text{SO}(5) & \begin{array} {c} 
\frac{1}{12} L \left(c_1^2-3 c_1 L+c_2+2 L^2\right)
\end{array} & \begin{array} {c} 
-\frac{1}{3} L \left(2 c_1^2-24 c_1 L+5 c_2+46 L^2\right)
\end{array} & \begin{array} {c} 
-\frac{1}{2} L \left(3 c_1^2+31 c_1 L-17 c_2-98 L^2\right)
\end{array} \\ 
\hline 
\text{SO}(6) & \begin{array} {c} 
\frac{1}{12} L \left(c_1^2-3 c_1 L+c_2+2 L^2\right)
\end{array} & \begin{array} {c} 
-\frac{1}{3} L \left(2 c_1^2-18 c_1 L+5 c_2+28 L^2\right)
\end{array} & \begin{array} {c} 
-\frac{1}{2} L \left(3 c_1^2+23 c_1 L-17 c_2-58 L^2\right)
\end{array} \\ 
\hline 
\end{array}
$}
\end{center}
\caption{Holomorphic genera.  To ease the notation, we  omit the degree $\int$ in the entries of the table. 
The SO($4$)-model has reducible Kodaira fibers supported on smooth divisors of classes $T$ and $S=4L-T$.
By definition, $L=c_1(\mathscr{L})$ and  $c_i$ denotes the $i$th Chern class of the base of the fibration. The holomorphic Euler characteristic $\chi_0(Y)$ is equal to $\chi_0(W,\mathscr{O}_W)$ where $W$ is the divisor defined by $\mathscr{L}$ in 
the base.}
\label{Table.HolomorphicECG}
\end{table}
\clearpage 

\begin{table}[htb!]
\begin{center}
\renewcommand{\arraystretch}{2}
\scalebox{1}{$
\begin{array}{|c|c|c|}
\hline 
\text{Type}  & \hspace{1cm}\int_Y p_2(TY) & \hspace{1cm} \int_Y p_1^2(TY)  \\\hline
\text{SO}(3) & \begin{array}{c} 
24L (-c_1^2 + 2 c_2 + 15 L^2)
\end{array} & \begin{array}{c}
48L (-c_1^2 +2 c_2 +5 L^2)
\end{array} \\
\hline 
\text{PSU}(3) & \begin{array}{c} 
24L (-c_1^2 + 2 c_2 + 8 L^2)
\end{array} & \begin{array}{c}
48L (-c_1^2 + 2 c_2 + 3 L^2)
\end{array} \\
\hline 
\text{SO}(4) & \begin{array}{c} 
24L (-c_1^2 + 2 c_2 + 15 L^2)-56 L T (4 L - T)
\end{array} & \begin{array}{c}
48L (-c_1^2 +2 c_2 +5 L^2)-32 L T (4 L - T)
\end{array} \\
\hline 
\text{SO}(5) & \begin{array}{c} 
8L \left(-3 c_1^2+6 c_2+17 L^2\right)
\end{array} & \begin{array}{c}
16L \left(-3 c_1^2+6 c_2+7 L^2\right)
\end{array} \\
\hline 
\text{SO}(6) & \begin{array}{c} 
8L \left(-3 c_1^2+6 c_2+10 L^2\right)
\end{array} & \begin{array}{c}
16 L \left(-3 c_1^2+6 c_2+5 L^2\right)
\end{array} \\
\hline 
\end{array}
$}
\end{center}
\caption{Pontryagin numbers. To ease the notation, we omit the degree $\int$ in the entries of the table. 
The SO($4$)-model has reducible Kodaira fibers supported on smooth divisors of classes $T$ and $S=4L-T$.
By definition, $L=c_1(\mathscr{L})$ and  $c_i$ denotes the $i$th Chern class of the base of the fibration.} 
\label{Table.PontryaginG}
\end{table}

\begin{table}[hb!]
\begin{center}
\renewcommand{\arraystretch}{1.6}
\scalebox{1}{$
\begin{array}{|c|c|c|c|}
\hline 
\text{Type} & 192 X_8=\int_Y( p_1^2-4p_2) & 45 \sigma=45\int_Y L_2= \int_Y(7p_2-p_1^2) & 5760\int_Y \hat{\text{A}}_2=\int_Y(7 p_1^2-4 p_2) \\\hline
\text{SO}(3) & \begin{array} {c} 
48L (c_1^2 - 2 c_2 - 25 L^2)
\end{array} & \begin{array} {c} 
120L (-c_1^2 + 2 c_2 + 19 L^2)
\end{array} & \begin{array} {c} 
240L (-c_1^2 + 2 c_2 + L^2)
\end{array} \\
\hline 
\text{PSU}(3) & \begin{array} {c} 
48L (c_1^2 - 2 c_2 - 13 L^2)
\end{array} & \begin{array} {c} 
120L (-c_1^2 + 2 c_2 + 10 L^2)
\end{array} & \begin{array} {c} 
240L (-c_1^2 + 2 c_2 + L^2)
\end{array} \\
\hline 
\text{SO}(4) & \begin{array} {c} 
48L (c_1^2 - 2 c_2 - 25 L^2)\\
 +192LT (4 L- T)
\end{array} & \begin{array} {c} 
120L (-c_1^2 + 2 c_2 + 19 L^2)\\
 -360LT(4 L- T)
\end{array} & \begin{array} {c} 
240L (-c_1^2 + 2 c_2 + L^2)
\end{array} \\
\hline 
\text{SO}(5) & \begin{array} {c} 
-48L \left(-c_1^2+2 c_2+9 L^2\right)
\end{array} & \begin{array} {c} 
120L \left(-c_1^2+2 c_2+7 L^2\right)
\end{array} & \begin{array} {c} 
240L (-c_1^2 + 2 c_2 + L^2)
\end{array} \\
\hline 
\text{SO}(6) & \begin{array} {c} 
-48L \left(-c_1^2+2 c_2+5 L^2\right)
\end{array} & \begin{array} {c} 
120L \left(-c_1^2+2 c_2+4 L^2\right)
\end{array} & \begin{array} {c} 
240L (-c_1^2 + 2 c_2 + L^2)
\end{array} \\
\hline 
\end{array}
$}
\end{center}
\caption{Characteristic invariants: the anomaly invariant $X_8$, the signature $\sigma$, and the index of the $\hat{\text{A}}$-genus. 
To ease the notation, we abuse notation and omit the degree $\int$ in the entries of the table. 
The SO($4$)-model has reducible Kodaira fibers supported on smooth divisors of classes $T$ and $S=4L-T$. By definition, $L=c_1(\mathscr{L})$ and  $c_i$ denotes the $i$th Chern class of the base of the fibration.}
\label{Table.Pontryagin2G}
\end{table}

\section*{Acknowledgments}
M.J.K. would like to thank Simons workshop 2018, Strings 2018, and String-math 2018 for their hospitality. 
M.E. is supported in part by the National Science Foundation (NSF) grant DMS-1701635  ``Elliptic Fibrations and String Theory''.
M.J.K. would like to acknowledge a partial support from NSF grant PHY-1352084.

\end{document}